\documentclass[11pt,a4paper]{article}

\usepackage[dvips]{graphicx}
\usepackage{amsmath}
\usepackage{amsthm}
\usepackage{amssymb}
\usepackage{verbatim}
\usepackage{algorithmic}
\usepackage{algorithm}
\usepackage{enumerate}
\usepackage{framed}
\usepackage{marginnote}
\usepackage{color}

\usepackage{algorithm}
\usepackage{algorithmic}

\usepackage{hyperref}

\newtheorem{theorem}{Theorem}
\newtheorem{definition}[theorem]{Definition}
\newtheorem{lemma}[theorem]{Lemma}

\newtheorem{corollary}[theorem]{Corollary}
\newtheorem{myclaim}{Claim}

\newtheorem{observation}{Observation}[section]

\newcommand{\R}{\mathbb{R}}
\newcommand{\Rp}{\mathbb{R}_{>0}}
\newcommand{\Zp}{\mathbb{Z}_{>0}}

\newcommand{\setN}{\mathbb{N}}
\newcommand{\I}{\mathcal{I}}

\newcommand{\T}{\mathcal{T}}

\newcommand{\Ce}{C_{\varepsilon}}
\newcommand{\xo}{x^*}

\newcommand{\x}{\bar{x}}

\newcommand{\KK}{\mathcal{K}}

\newcommand{\mink}{\textsc{min-K}}
\newcommand{\maxk}{\textsc{max-K}}

\newcommand{\rmink}{\textsc{MinLP}}
\newcommand{\ik}{\textsc{IK}}
\newcommand{\iik}{\textsc{IIK}}
\newcommand{\ufp}{\textsc{UFP}}

\newcommand{\conv}[1]{\textsc{conv}(#1)}

\newcommand{\qedd}{\hfill $\blacksquare$}

\usepackage{mathtools}

\DeclarePairedDelimiter\floor{\lfloor}{\rfloor}

\title{A PTAS for the \\ Time-Invariant Incremental Knapsack problem}
\author{Yuri Faenza\thanks{IEOR, Columbia University, USA. Email: yf2414@columbia.edu} \and Igor Malinovi\'c\thanks{DISOPT, EPFL, Switzerland. Email: igor.malinovic@epfl.ch}}
\date{\today}
\begin{document}
	\maketitle

	
	\begin{abstract}
		
		The Time-Invariant Incremental Knapsack problem (\iik) is a generalization of Maximum Knapsack to a discrete multi-period setting. At each time, capacity increases and items can be added, but not removed from the knapsack. The goal is to maximize the sum of profits over all times. \iik \ models various applications including specific financial markets and governmental decision processes. \iik \ is strongly NP-hard \cite{Bienstock13} and there has been work \cite{Bienstock13,iik17,Hartline08,Sharp07,Ye16} on giving approximation algorithms for some special cases. In this paper, we settle the complexity of \iik \ by designing a PTAS based on rounding a disjuncive formulation, and provide several extensions of the technique.
	
		
	\bigskip

			
	\end{abstract}
	
	\section{Introduction}
	
	Knapsack problems are among the most fundamental and well-studied in discrete optimization. Some variants forego the development of modern optimization theory, dating back to 1896 \cite{Mathews1896}. The best known representative is arguably \emph{Maximum Knapsack} (\maxk): given a set of items with specified profits and weights, and a threshold, find a most profitable subset of items whose total weight does not exceed the threshold. \maxk \ is NP-complete \cite{NPComplete}, while admitting a  \textit{fully polynomial-time approximation scheme} (FPTAS) \cite{Ibarra}. Many classical algorithmic techniques including greedy, dynamic programming, backtracking/branch-and-bound have been studied by means of solving this problem, see e.g. \cite{knap-book}. The algorithm of Martello and Toth \cite{Martello} has been known to be the fastest in practice for exactly solving knapsack instances \cite{KnapTest}.
	
		In order to model scenarios arising in real-world applications, more complex knapsack problems have been introduced (see \cite{knap-book} for a survey) and recent works studied extensions of classical combinatorial optimization problems to multi-period settings, see e.g. \cite{Hartline08,Sharp07,Skutella08}. 
		%
		At the intersection of those two streams of research, Bienstock et al. \cite{Bienstock13} proposed a generalization of a \maxk \ to a multi-period setting that they dubbed \emph{Time-Invariant Incremental Knapsack} (\iik). In \iik, we are given a set of items $[n]$ with profits $p : [n] \rightarrow \Rp$ and weights $w : [n] \rightarrow \Rp$ and a knapsack with non decreasing capacity $b_t$ over time $t \in [T]$. We can add items at each time as long as the capacity constraint is not violated, and once inserted, an item cannot be removed from the knapsack. The goal is to maximize the total profit,  which is defined to be the sum, over $t \in [T]$, of profits of items in the knapsack at time $t$. 
		
		\iik \ 
		models a scenario where available resources (e.g. money, labour force) augment over time in a predictable way, allowing to grow our portfolio. Take e.g. a bond market with an extremely low level of volatility, where all coupons render profit only at their common maturity time $T$ (\emph{zero-coupon} bonds) and an increasing budget over time that allows buying more and more (differently sized and priced) packages of those bonds. For  variations of \maxk \ that have been  used to model financial problems, see \cite{knap-book}. A different application arises in government-type decision processes, where items are assets of public utility (schools, parks, etc.) that can be built at a given cost and give a yearly benefit (both constant over the years), and the community will profit each year those assets are available. %
		
\medskip 
	
	\noindent \textbf{Previous work on \iik.} 	Although the first publication on \iik \ appeared just very recently \cite{iik17}, it was previously studied in \cite{Bienstock13} and several PhD theses \cite{Hartline08,Sharp07,Ye16}. Here we summarize all those results. In \cite{Bienstock13}, \iik \ is shown to be strongly NP-hard and an instance showing that the natural LP relaxation has unbounded integrality gap is provided. In the same paper, a PTAS is designed for $T=O(\log n)$. This improves over \cite{Sharp07}, where a PTAS for the special case $p=w$ is given when $T$ is a constant. Again when $p=w$, a 1/2-approximation algorithm for generic $T$ is provided in \cite{Hartline08}. Results from \cite{Ye16} can be adapted to give an algorithm that solves \iik \ in time polynomial in $n$ and of order $(\log T)^{O(\log T)}$ for a fixed approximation guarantee $\varepsilon$ \cite{personal-jay}. The authors in \cite{iik17} provide an alternative PTAS for \iik\ with constant $T$, and a $1/2$-approximation for arbitrary $T$ with under the assumption that every item alone fits into the knapsack at $t=1$.  

	\medskip 
	
	\noindent \textbf{Our contributions.} In this paper, we give an algorithm for computing a $(1-\varepsilon)$-approximated solution for \iik \ that depends polynomially on the number $n$ of items and, for any fixed $\varepsilon$, also polynomially on the number of times $T$. In particular, our algorithm provides a PTAS for \iik, regardless of $T$.

	\begin{theorem}\label{thr:IIK}
		There exists an algorithm that, when given as input $\varepsilon \in \R_{>0}$ and an instance ${\cal I}$ of \iik \ with $n$ items and $T\geq 2$ times, produces a $(1-\varepsilon)$-approximation to the optimal solution of ${\cal I}$ in time $O(T^{h(\varepsilon)} \cdot n f_{LP}(n))$. Here $f_{LP}(m)$ is the time required to solve a linear program with $O(m)$ variables and constraints, and $h: \R_{>0} \rightarrow \R_{\geq 1}$ is a function depending on $\varepsilon$ only. In particular, there exists a PTAS for \iik.
	\end{theorem}

	Theorem \ref{thr:IIK} dominates all previous results on \iik \ \cite{Bienstock13,iik17,Hartline08,Sharp07,Ye16}  and, due to the hardness results in \cite{Bienstock13}, settles the complexity of the problem. Interestingly, it is based on designing a disjunctive formulation -- a tool mostly common among integer programmers and practitioners\footnote{See Appendix \ref{app:disj-progr} for a discussion on disjunctive programming.} -- and then rounding the solution to its linear relaxation with a greedy-like algorithm. We see Theorem \ref{thr:IIK} as an important step towards the understanding of the complexity landscape of knapsack problems over time. Theorem \ref{thr:IIK} is proved in Section \ref{sec:IIK}: see the end of the current section for a sketch of the techniques we use and a detailed summary of Section \ref{sec:IIK}. In Section \ref{sec:ext}, we show some extensions of Theorem \ref{thr:IIK} to more general problems. 


\medskip

	\noindent {\bf{Related work on other knapsack problems.}} 
	\cite{Bienstock13} discusses the relation between \iik \ and the generalized assignment problem (GAP), highlighting the differences between those problems. In particular, there does not seem to be a direct way to apply to \iik \ the $(1-1/e-\varepsilon)$ approximation algorithm \cite{Fleischer06} for GAP. 
	Other generalizations of \maxk \ related to \iik, but whose current solving approaches do not seem to extend, are the multiple knapsack (MKP) and unsplittable flow on a path (\ufp) problems. In Appendix \ref{app:ufp} we discuss those problems in order to highlight the new ingredients introduced by our approach. 
	
\medskip

\noindent {\bf{The basic techniques.}} 
	In order to illustrate the ideas behind the proof of Theorem \ref{thr:IIK}, let us first recall one of the PTAS for the classical \maxk \ with capacity $\beta$, $n$ items, profit and weight vector $p$ and $w$ respectively. Recall the greedy algorithm for knapsack:
	\begin{enumerate}
		
		\item Sort items so that $\frac{p_{1}}{w_{1}} \geq \frac{p_{2}}{w_{2}} \geq \cdots
		\geq \frac{p_{n}}{w_{n}}$.
		\medskip
		\item Set $\x_i=1$ for $i=1,\dots,\bar \imath$, where $\bar \imath$ is the maximum integer s.t. $\sum\limits_{1\leq i \leq \bar \imath} w_i \leq \beta$.
		
	\end{enumerate}
	
	\noindent It is well-known that $p^T \x \geq p^T x^*-\max_{i\geq \bar \imath +1}p_i$, where $x^*$ is the optimal solution to the linear relaxation. A PTAS for \maxk \ can then be obtained as follows: ``guess'' a set $S_0$ of $\frac{1}{\varepsilon}$ items with $w(S_0)\leq \beta$ and consider the ``residual'' knapsack instance ${\cal I}$ obtained removing items in $S_0$ and items $\ell$ with $p_\ell > \min_{i \in S_0} p_i$, and setting the capacity to $\beta-w(S_0)$. Apply the greedy algorithm to ${\cal I}$ as to obtain solution $S$. Clearly $S_0 \cup S$ is a feasible solution to the original knapsack problem. The best solutions generated by all those guesses can be easily shown to be a $(1-\varepsilon)$-approximation to the original problem.
	
	Recall that \iik \ can be defined as follows.
	
		\begin{equation}\label{eq:IIK}
		\begin{array}{ll}
			\max & \quad \sum\limits_{t \in [T]} p^Tx_t \\
			\text{s.t.} & \quad  w^Tx_t \leq b_t \quad \forall t \in [T] \\
			\ & \quad x_{t} \leq x_{t+1} \quad \forall t \in [T-1] \\
			\ & \quad x_t \in \{0,1\}^n \quad \forall t \in [T].
		\end{array}
	\end{equation}
	
		By definition, $0 < b_t \leq b_{t+1}$ for $t \in [T-1]$. We also assume wlog that $1=p_1\geq p_2 \geq ... \geq p_n$.

	When trying to extend the PTAS above for \maxk \ to \iik, we face two problems. First, we have multiple times, and a standard guessing over all times will clearly be exponential in $T$. Second, when inserting an item into the knapsack at a specific time, we are clearly imposing this decision on all times that succeed it, and it is not clear a priori how to take this into account.

	We solve these issues by proposing an algorithm that, in a sense, still follows the general scheme of the greedy algorithm sketched above: after some preprocessing, guess items (and insertion times) that give high profit, and then fill the remaining capacity with an LP-driven integral solution. However, the way of achieving this is different from the PTAS above. In particular, some of the techniques we introduced are specific for \iik \ and not to be found in methods for solving non-incremental knapsack problems.
	 
	\medskip 
	\noindent \textbf{An overview of the algorithm}:
	
	\begin{enumerate}
	\item[(i)] \emph{Sparsification and other simplifying assumptions}. We first show that by losing at most a $2\varepsilon$ fraction of the profit, we can assume the following (see Section \ref{sec:reduction}): item $1$, which has the maximum profit, is inserted into the knapsack at some time; the capacity of the knapsack only increases and hence the insertion of items can only happen at $J=O(\frac{1}{\varepsilon}\log T)$ times (we call them \emph{significant}); and the profit of each item is either much smaller than $p_1=1$ or it takes one of  $K=O(\frac{1}{\varepsilon}\log \frac{T}{\varepsilon})$ possible values (we call them \emph{profit classes}).
	
\medskip
	 
	\item[(ii)]\emph{Guessing of a stairway}. The operations in the previous step give a $J \times K$ grid of ``significant times'' vs ``profit classes'' with $O(\frac{1}{\varepsilon^2}\log^2\frac{T}{\varepsilon})$ entries in total. One could think of the following strategy: for each entry $(j,k)$ of the grid, guess how many items of profit class $k$ are inserted in the knapsack at time $t_j$. However, those entries are still too many to perform guessing over all of them. Instead, we proceed as follows: we guess, for each significant time $t_j$, \emph{which is the class} $k$ of maximum profit that has an element in the knapsack at time $t_j$. Then, for profit class $k$ and carefully selected profit classes ``close'' to $k$, we either \emph{guess exactly how many items} are in the knapsack at time $t_j$ or if these are \emph{at least} $\frac{1}{\varepsilon}$. Each of the guesses leads to a natural IP. The optimal solution to one of the IPs is an optimal solution to our original problem. Clearly, the number of possible guesses affects the number of the IPs, hence the overall complexity. We introduce the concept of ``stairway'' to show that these guesses are polynomially many for fixed $\epsilon$. See Section \ref{sec:iikdisj} for details. We remark that, from this step on, we substantially differ from the approach of \cite{Bienstock13}, which is also based on a disjunctive formulation. 
	
\medskip

	\item[(iii)] \emph{Solving the linear relaxations and rounding}. Fix an IP generated at the previous step, and let $x^*$ be the optimal solution of its linear relaxation. A classical rounding argument relies on LP solutions having a small number of fractional components. Unfortunately, $x^*$ is not as simple as that. However, we show that, after some massaging, we can control the entries of $x^*$ where ``most'' fractional components appear, and conclude that the profit of $\lfloor x^*\rfloor$ is close to that of $x^*$. See Section \ref{sec:round} for details. Hence, looping over all guessed IPs and outputting vector $\lfloor x^*\rfloor$ of maximum profit concludes the algorithm.
	\end{enumerate}
\bigskip

	
	
	
\noindent {\bf Assumption:} We assume that expressions $\frac{1}{\varepsilon}$, $(1+\varepsilon)^{j}$, $\ \log_{1+\varepsilon}{\frac{T}{\varepsilon}}$ and similar are to be rounded up to the closest integer. This is just done for simplicity of notation and can be achieved by replacing $\epsilon$ with an appropriate constant fraction of it, which will not affect the asymptotic running time.
	
	
\section{A PTAS for IIK}\label{sec:IIK}

	\subsection{Reducing IIK to special instances and solutions}\label{sec:reduction}
	
	Our first step will be to show that we can reduce \iik, without loss of generality, to solutions and instances with a special structure. The first reduction is immediate: we restrict to solutions where the highest profit item is inserted in the knapsack at some time. We call these \emph{$1$-in solutions}. This can be assumed by guessing which is the highest profit item that is inserted in the knapsack, and reducing to the instance where all higher profit items have been excluded. Since we have $n$ possible guesses, the running time is scaled by a factor $O(n)$.	
	\begin{observation}\label{obs:1-in}
		Suppose there exists a function $f: \setN \times \setN \times \R_{> 0}$ such that, for each $n,T \in \setN$, $\varepsilon> 0$, and any instance of \iik \ with $n$ items and $T$ times, we can find a $(1 - \varepsilon)$-approximation to a $1$-in solution of highest profit in time $f(n,T,\varepsilon)$. Then we can find a $(1-\varepsilon)$-approximation to any instance of \iik \ with $n$ items and $T$ times in time $O(n) \cdot f(n,T,\varepsilon)$.
	\end{observation}
	
	Now, let $\cal I$ be an instance of \iik \ with $n$ items, let $\varepsilon > 0$. We say that $\cal I$ is \emph{$\varepsilon$-{well-behaved}} if it satisfies the following properties. 	
	\begin{enumerate}[($\varepsilon$1)]
		\item\label{profit-constraint} For all $i \in [n]$, one has $p_i = (1+\varepsilon)^{-j}$ for some $j \in \{0,1, \dots , \log_{1+\varepsilon} \frac{T}{\varepsilon} \} $, or $p_i \leq \frac{\varepsilon}{T}$.
		\item\label{time-constraint} $b_t=b_{t-1}$ for all $t \in [T]$ such that $(1+\varepsilon)^{j-1}<T-t+1<(1+\varepsilon)^{j}$ for some \\ $j \in \{0,1, \dots , \log_{1+\varepsilon} T \}$, where we set $b_0=0$.
	\end{enumerate}
	\begin{figure}[ht] 
	\center
	\includegraphics[scale=0.8]{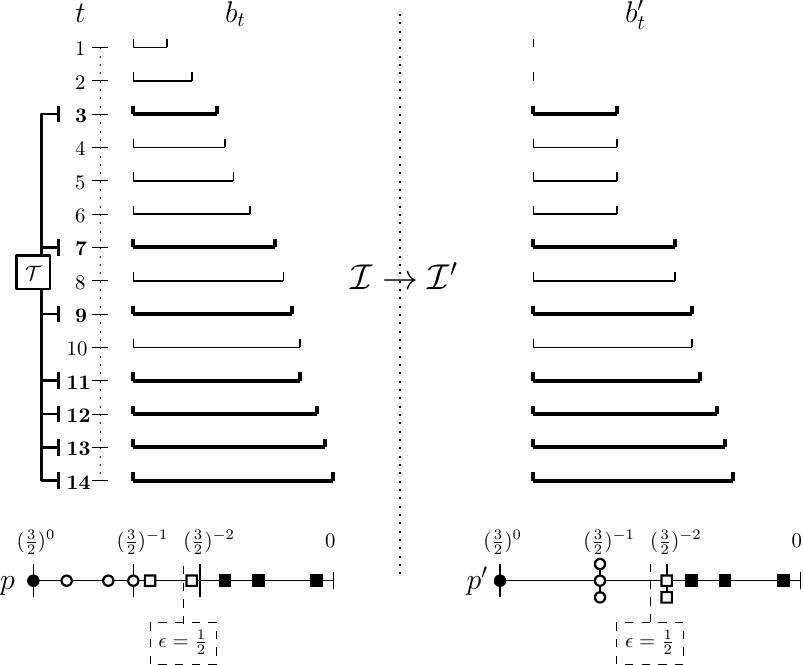}
	\caption{An example of obtaining an $\varepsilon$-well-behaved instance for $\varepsilon=\frac{1}{2}$ and $T=14$.}
	\label{fig:well-behaved}
\end{figure}

	See Figure \ref{fig:well-behaved} for an example. 
	Note that condition $(\varepsilon$\ref{time-constraint}$)$ implies that the capacity can change only during the set of times $ \T := \{ t \in [T]: t=T+1-{ (1 + \varepsilon )^j } \text{ for some } j \in \setN \}$, with $|\T|=O(\log_{1+\varepsilon} T)$. $\T$ clearly gets sparser as $t$ becames smaller. Note that for $T$ not being a degree of $(1 + \varepsilon)$ there will be a small fraction of times $t$ at the beginning with capacity $0$; see Figure \ref{fig:well-behaved}. 

	Next theorem implies that we can, wlog, assume that our instances are $\varepsilon$-well-behaved (and our solutions are $1$-in). 	
	
	\begin{theorem}\label{thr:well-behaved}
		Suppose there exists a function $g: \setN \times \setN \times \R_{> 0}$ such that, for each $n,T \in \setN$, $\varepsilon> 0$, and any $\varepsilon$-well-behaved instance of \iik \ with $n$ items and $T$ times, we can find a $(1 - 2\varepsilon)$-approximation to a $1$-in solution of highest profit in time $g(n,T,\varepsilon)$. Then we can find a $(1-4\varepsilon)$-approximation to any instance of \iik \ with $n$ items and $T$ times in time $O(T+n(n+g(n,T,\varepsilon))$. 
	\end{theorem}
	
	Fix an \iik \ instance ${\cal I}$. The reason why we can restrict ourselves to finding a $1$-in solution is Observation \ref{obs:1-in}. Denote with ${\cal I'}$ the instance  with $n$ items having the same weights as in ${\cal I}$, $T$ times, and the other parameters defined as follows:
		\begin{itemize}
			\item     
			For $i \in [n]$, if $(1+\varepsilon)^{-j} \leq {p_i}< (1+\varepsilon)^{-j+1}$ for some $j \in \{0,1, \dots , \log_{1+\varepsilon} \frac{T}{\varepsilon} \} $, set $p_{i}' :=(1+\varepsilon)^{-j}$; otherwise, set $p_i':={p_i}$ . Note that we have $1=p'_1\geq p'_2 \geq ... \geq p'_{n}$.
			\item
			For $t \in [T]$ and  $(1+\varepsilon)^{j-1} < T-t+1 \leq (1+\varepsilon)^{j} $
			for some $j \in \{0,1, \dots , \log_{1+\varepsilon} T \}$, set $b_t':=b_{T-(1+\varepsilon)^j+1}$, with $b_0':=0$.
		\end{itemize}

		One easily verifies that $\I'$ is $\varepsilon$-well-behaved. Moreover, $b'_t\leq b_t$ for all $t\in [T]$ and $\frac{p_i}{1+\varepsilon} \leq p'_i\leq p_i$ for $i \in [n]$, so we deduce: \begin{myclaim}\label{obs:eps-well-behaved-inside}
			Any solution $\bar x$ feasible for $\I'$ is also feasible for $\I$, and $p(\bar x)\geq p'(\bar x)$. 
		\end{myclaim}
		%
		\noindent We also prove the following.
		
		\begin{myclaim}\label{claim:eps-well-behaved-inside}
			Let $x^*$ be a $1$-in feasible solution of highest profit for $\I$. There exists a $1$-in feasible solution $x'$ for $\I'$ such that $p'(x') \geq (1-\varepsilon)^2 p(x^*)$. 
		\end{myclaim}
		
		\begin{proof}

Define $x' \in \{0,1\}^{Tn}$ as follows:
		$$\begin{array}{lll}
		x_t' := x^*_{T-(1+\varepsilon)^j+1} \quad \hbox{ if } (1+\varepsilon)^{j-1} < T- t +1  \leq (1+\varepsilon)^{j} \\ \qquad \hbox{ for } j \in \{0,1, \dots , \log_{1+\varepsilon} T \}, \hbox{with } x^*_0=0. 
		\end{array}$$
		
		In order to prove the claim we first show that $x'$ is a feasible $1$-in solution for $\I'$. Indeed, it is $1$-in, since by construction $x'_{T,1}=x^*_{T,1}=1$. It is feasible, since 
		for $t$ such that $(1+\varepsilon)^{j-1} < T- t +1  \leq (1+\varepsilon)^{j}, \ j \in \{0,1, \dots , \log_{1+\varepsilon} T \}$ we have $$w^T x'_t=w^T x^*_{T-(1+\varepsilon)^j+1}\leq b_{T-(1+\varepsilon)^j+1} = b'_t.$$ 
		Comparing $p'(x')$ and $p(x^*)$ gives 
		$$\begin{array}{lllllll} 
		p'(x')  & \geq & \sum\limits_{t \in [T]} \sum\limits_{i \in [n]} p_i' x_{t,i}' & = &  \sum\limits_{i \in [n]} (T-t_{i, \min}(x') + 1) p_i' \\
		& \geq & \sum\limits_{i \in [n]} \frac{1}{1 + \varepsilon} (T-t_{i, \min}(x^*) + 1) p_i' & \geq & \sum\limits_{i \in [n]} \frac{1}{(1 + \varepsilon)^2} (T-t_{i, \min}(x^*) + 1) p_i \\
		& = & (\frac{1}{1 + \varepsilon})^2 p(x^*) & \geq & (1 - \varepsilon)^2 p(x^*),
		\end{array}$$
		\noindent where $t_{i, \min}(v):= \min\{t \in [T] : \ v_{t,i} = 1 \}$ for $v \in \{0,1\}^{Tn}$. 
		\end{proof}

\noindent	\textbf{Proof of Theorem \ref{thr:well-behaved}.} Let $\hat x$ be a $1$-in solution of highest profit for ${\cal I}'$ and $\bar{x}$ is a solution to $\I'$ that is a $(1-\varepsilon)$-approximation to $\hat x$. Claim \ref{obs:eps-well-behaved-inside} and Claim \ref{claim:eps-well-behaved-inside} imply that $\bar x$ is feasible for ${\cal I}$ and we deduce:
		$$p(\bar{x}) \geq p'(\bar x) \geq (1-2\varepsilon) p'(\hat x) 
		{\geq} (1-2\varepsilon) p'(x') \geq (1-2\varepsilon)(1-\varepsilon)^2 p(x^*) \geq (1 - 4 \varepsilon) p(x^*).$$

		In order to compute the running time, it is enough to bound the time required to produce ${\cal I}'$. Vector $p'$ can be produced in time $O(n)$, while vector $b'$ in time $T$. Moreover, the construction of the latter can be performed before fixing the highest profit object that belongs to the knapsack (see Observation \ref{obs:1-in}). The thesis follows.
		
		\qedd		
		

	\subsection{A disjunctive relaxation} \label{sec:iikdisj}
	
	Fix $\varepsilon>0$. Because of Theorem \ref{thr:well-behaved}, we can assume that the input instance $\cal I$ is $\varepsilon$-well-behaved. We call all times from $\T$ \emph{significant}. Note that a solution over the latter times can be naturally extended to a global solution by setting $x_t=x_{t-1}$ for all non-significant times $t$. 
	We denote significant times by $t(1) < t(2) < \dots < t(|\T|)$.	In this section, we describe an IP over feasible $1$-in solutions of an $\varepsilon$-well-behaved instance of \iik.
	The feasible region of this IP is the union of different regions, each corresponding to a partial assignment of items to significant times. 
	In Section \ref{sec:round} we give a strategy to round an optimal solution of the LP relaxation of the IP to a feasible integral solution with a $(1-2\varepsilon)$-approximation guarantee. Together with Theorem \ref{thr:well-behaved} (taking $\varepsilon'=\frac{\varepsilon}{4}$), this implies Theorem \ref{thr:IIK}.
	
	
	In order to describe those partial assignments, we introduce some additional notation. We say that items having profit $(1+\varepsilon)^{-k}$ for $k \in [ \log_{1+\varepsilon} \frac{T}{\varepsilon}] $, belong to \emph{profit class $k$}. Hence bigger profit classes correspond to items with smaller profit. All other items are said to belong to the \emph{small} profit class. Note that there are $O(\frac{1}{\varepsilon} \log \frac{T}{\varepsilon})$ profit classes (some of which could be empty).	Our partial assignments will be induced by special sets of vertices of a related graph called \emph{grid}.

	\begin{definition}
		Let $J \in \Zp, K \in \mathbb{Z}_{\geq 0}$, a grid of dimension $J \times (K+1)$ is the graph $G_{J,K}=([J] \times [K]_0,E)$, where 
		\begin{align*}
			E & :=\{\{u,v\}: \ u, v \in [J] \times [K]_0, \ u=(j, k) \\
			\ & \quad \quad \quad \quad \quad \ \text{ and either } v=(j+1, k) \text{ or } v=(j, k+1) \ \}.
		\end{align*}
	\end{definition}
	

	\begin{definition}\label{def:stairway}
		Given a grid $G_{J,K}$, we say that 
	$$S:=\{(j_1, k_1), (j_2, k_2), \dots , (j_{|S|}, k_{|S|})\}\subseteq  V(G_{J,K})$$
		 is a \emph{stairway} if $
		j_h > j_{h+1} \text{ and } k_h < k_{h+1}  \text{ for all } h \in [|S|-1]$.
		
	\end{definition}

	\begin{lemma}\label{lem:stairway}
		There are at most $2^{K+J+1}$ distinct stairways in $G_{J,K}$.
	\end{lemma}
	
	\begin{proof}
	The first coordinate of any entry of a stairway can be chosen among $J$ values, the second coordinate from $K+1$ values. By Definiton \ref{def:stairway}, each stairway correspond to exactly one choice of sets $J_1\subseteq [J]$ for the first coordinates and $K_1\subseteq [K]_0$ for the second, with $|K_1|=|J_1|$.	
		
	\end{proof}

	
	
	
	
	Now consider the grid graph with $J:=|\T|=\theta(\frac{1}{\varepsilon} \log T)$, $K= \log_{1+\varepsilon} \frac{T}{\varepsilon}$, and a stairway $S$ with $k_1 = 0$. See Figure \ref{fig:stairway} for an example. This corresponds to a partial assignment that can be informally described as follows. Let $(j_h,k_h)\in S$ and $t_h := t(j_h)$. In the corresponding partial assignment no item belonging to profit classes $k_h \leq k < k_{h+1}$ is inside the knapsack at any time  $t < t_h$, while the first time an item from profit class $k_h$ is inserted into the knapsack is at time $t_h$ (if $j_{|S|}>1$ then the only items that the knapsack can contain at times $1,\dots,t_{|S|}-1$ are the items from the small profit class). Moreover, for each $h \in [|S|]$, we focus on the family of profit classes $\KK_h := \{ k \in [K]: \ k_h \leq k \leq k_h + \Ce\}$ with $\Ce={\log_{1+\varepsilon} \frac{1}{\varepsilon}}$. For each $k \in \KK_h$ and every (significant) time $t$ in the set $\T_h := \{ t \in \T: \ t_{h-1} < t \leq t_h  \}$, we will either specify exactly the number of items taken from profit class $k$ at time $t$, or impose that there are at least $\frac{1}{\varepsilon}+1$ of those items (this is established by map $\rho_h$ below). Note that we can assume that the items taken within a profit class are those with minimum weight: this may exclude some feasible $1$-in solutions, but it will always keep at least a feasible $1$-in solution of maximum profit. 
	No other constraint is imposed.
		
	More formally, set $k_{|S|+1}=K+1$ and for each $h=1,\dots, |S|$:
	\begin{enumerate}[i)]
		\item Set $x_{t,i}=0$ for all $t \in [t_h-1]$ and each item $i$ in a profit class $k \in [k_{h+1}-1]$.
		\item Fix a map $\rho_h: \T_h \times \KK_h \rightarrow \{0,1, \dots , {\frac{1}{\varepsilon}} +1 \}$ such that for all $t \in \T_h$ one has $\rho_h(t, k_h) \geq 1$ and $ \rho_h(\bar{t},k) \geq \rho_h(t,k), \ \forall (\bar{t},k) \in \T_h \times \KK_h, \ \bar{t} \geq t$. 
		
	\end{enumerate}
	
	

	
	Additionally, we require $\rho_h(\bar{t},k) \geq \rho_{h+1}(t,k)$ for all $h \in [|S|-1], \ k \in \KK_h \cap \KK_{h+1}, \ \bar{t} \in \T_{h}, \ t \in \T_{h+1}$. Thus, we can merge all $\rho_h$ into a function $\rho: \cup_{h \in [|S|]} (\T_h \times \KK_h) \rightarrow \{0,1, \dots , {\frac{1}{\varepsilon}} +1 \}$. For each profit class $k \in [K]$ we assume that items from this class are $I_k=\{1(k),\dots , |I_k|(k) \}$, so that $w_{1(k)} \leq w_{2(k)} \leq \dots \leq w_{{|I_k|}(k)}$. Based on our choice $(S, \rho)$ we define the polytope: 
$$\begin{array}{llll}
		P(S,\rho)=\{x \in \R^ {Tn}: &&  w^Tx_t \leq b_t & \forall t \in [T] \\
		&& x_{t} \leq x_{t+1} & \forall t \in [T-1] \\
		&& 0 \leq x_t \leq 1 & \forall t \in [T] \\
		& \forall h \in [|S|]:&  \\
		&& x_{t,i(k)}=0, & \forall t < t_h, \ \forall k < k_{h+1}, \ \forall i(k) \in I_k \\
		&&  x_{t,i(k)}=1, & \forall t \in \T_h, \ \forall k \in \KK_h, \ \forall i(k): i \leq \rho(t,k) \\
		&& x_{t,i(k)}=0,  & \forall t \in \T_h, \ \forall k \in \KK_h: \rho(t,k) \leq {\frac{1}{\varepsilon}}, \\
		&&& \forall i(k): i > \rho(t,k) \}. \ \
	\end{array}$$

\begin{figure}[h!t]
	\center
	\includegraphics[scale=1.0]{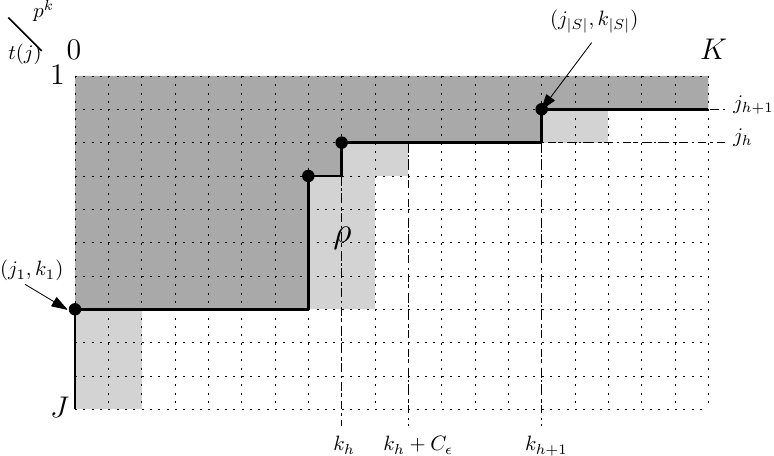}
	\caption{An example of a stairway $S$, given by thick black dots. Entries $(j,k)$ lying in the light grey area are those for which a value $\rho$ is specified. No item corresponding to the entries in the dark grey area is taken, except on the boundary in bold.}
	\label{fig:stairway}
\end{figure}	

	The linear inequalities are those from the IIK formulation. The first set of equations impose that, at each time $t$, we do not take any object from a profit class $k$, if we guessed that the highest profit object in the solution at time $t$ belongs to a profit class $k'>k$ (those are entries corresponding to the dark grey area in Figure \ref{fig:stairway}). The second set of equations impose that for each time $t$ and class $k$ for which a guess $\rho(t,k)$ was made (light grey area in Figure \ref{fig:stairway}), we take the $\rho(t,k)$ items of smallest weight. As mentioned above, this is done without loss of generality: since profits of objects from a given profit class are the same, we can assume that the optimal solution insert first those of smallest weight. The last set of equations imply that no other object of class $k$ is inserted in time $t$ if $\rho(t,k)\leq {\frac{1}{\varepsilon}}$.  
	
	Note that some choices of $S, \rho$ may lead to empty polytopes. Fix $S, \rho$, an item $i$ and some time  $t$. If, for some $t'\leq t$, $x_{t',i}=1$ explicitly appears in the definition of $P(S,\rho)$ above, then we say that $i$ is \emph{$t$-included}. Conversely, if $x_{\bar{t},i}=0$ explicitly appears for some $\bar{t} \geq t$, then we say that $i$ is \emph{$t$-excluded}.

	\begin{theorem}\label{thr:disjunction-iik}
		Any optimal solution of
		$$\max \sum_{t \in [T]}p_{t}^T x_{t} \; \; \hbox{ s.t. }\; \; \ x\in \left(\cup_{S,\rho} P(S, \rho) \right) \cap \{0,1\}^ {Tn} $$ 
		is a $1$-in solution of maximum profit for $\cal I$. Moreover, the the number of constraints of the associated LP relaxation is at most $nT^{f(\varepsilon)}$ for some function $f : \R_{>0}\rightarrow \R_{>0}$ depending on $\varepsilon$ only.	
	\end{theorem}
	
	\begin{proof}
	Note that one of the choices of $(S,\rho)$ will be the correct one, i.e. it will predict the stairway $S$ associated to an optimal $1$-in solution, as well as the number of items that this solution takes for each entry of the grid it guessed. Then there exists an optimal solution that takes, for each time $t$ and class $k$ for which a guess $\rho(t,k)$ was made, the $\rho(t,k)$ items of smallest weight from this class, and no other object if $\rho(t,k) \leq {\frac{1}{\epsilon}}$. These are exactly the constraints imposed in $P(S,\rho)$. The second part of the statement follows from the fact that the possible choices of $(S, \rho)$ are 
		$$\begin{array}{ccrl}
		(\# \hbox{ stairways})& \cdot & (\# \hbox{ possible values in each entry of } \rho)& ^{( \# \hbox{ entries of a vector $\rho$ })}\\
		& & \multicolumn{1}{c}{=} \\
		2^{O(\frac{1}{\varepsilon}\log \frac{T}{\varepsilon})} & \cdot & 
		O(\frac{1}{\varepsilon}) & ^{
			O(\frac{1}{\varepsilon} \log \frac{T}{\varepsilon}) C_\varepsilon}\\
		& & \multicolumn{1}{c}{=} \\
		\multicolumn{4}{c}{(\frac{T}{\varepsilon})^{O(\frac{1}{\varepsilon})} \cdot ( \frac{T}{\varepsilon})^{O((\frac{1}{\varepsilon})^3)},}
		\end{array}
		$$
		
		\noindent and each $(S, \rho)$ has $g(\varepsilon)O(Tn)$ constraints, where $g$ depends on $\varepsilon$ only.	
	%
	\end{proof}

\subsection{Rounding}\label{sec:round}

By convexity, there is a choice of $S$ and $\rho$ given as in the previous section such that any optimal solution of 
\begin{equation}
\max \sum_{t \in [T]}p^T x_{t}  \; \; \hbox{ s.t. }\; \; \  x \in P(S, \rho)  \label{eq:opt-S-rho}
\end{equation} 
is also an optimal solution to 
$$\max \sum_{t \in [T]}p^T x_{t}  \; \; \hbox{ s.t. }\; \; \  x \in \text{conv}(\cup_{S,\rho} P(S, \rho)).$$ 
Hence, we can focus on rounding an optimal solution $\xo$ of \eqref{eq:opt-S-rho}. We assume that the items are ordered so that $\frac{p_1}{w_1} \geq \frac{p_2}{w_2} \geq \dots \geq \frac{p_n}{w_n}$. Moreover, let ${\cal I}^t$ (resp. ${\cal E}^t$) be the set of items from $[n]$ that are $t$-included (resp. $t$-excluded) for $t \in [T]$, and let $W_t:=w^Tx^*_t$.

	\begin{algorithm}[h!t]
	\caption{}
	\label{algo:x-opt}
	\begin{algorithmic} [1]
	\medskip
		\STATE{Set $\bar{x}_{0}=\boldsymbol{0}$.}
		\medskip
		\STATE{For $t=1,\dots,T$:
			\vspace{-2mm}		    
			\begin{enumerate}[(a)]
				\setlength{\itemsep}{0pt}
				
				\item Set $\bar{x}_t=\bar{x}_{t-1}$.
				\item\label{step:add-included} Set $\bar{x}_{t,i}=1$ for all $i \in {\cal I}^t$. 
				\item While {$W_t - w^T\bar{x}_t>0$:} 
				\vspace{-3mm}\begin{enumerate}[(i)]
				\setlength{\itemsep}{0pt}
					\item Select the smallest $i \in [n]$ such that $i \notin {\cal E}^t$ and $\bar{x}_{t,i}<1$.
					\item Set $\bar{x}_{t,i}=\bar{x}_{t,i} + \min \{1-\bar{x}_{t,i}, \frac{W_t - w^T\bar{x}_t }{w_i} \}$. 
				\end{enumerate}
			\end{enumerate}}
		\end{algorithmic}
	\end{algorithm}


\noindent Respecting the choices of $S$ and $\rho$, i.e. included/excluded items at each time $t$, Algorithm \ref{algo:x-opt} greedly adds objects into the knapsack, until the total weight is equal to $W_t$. Recall that in \maxk \ one obtains a rounded solution which differs from the fractional optimum by the profit of at most one item. Here the fractionality pattern is more complex, but still under control. In fact, as we show below, $\x$ is such that $\sum_{t \in [T]}p^T \x_{t} = \sum_{t \in [T]}p^T \xo_{t}$ and, for each $h \in [|S|]$ and $t \in [T]$ such that $t_h \leq t < t_{h-1}$, vector $\x_t$ has at most $|S|-h+1$ fractional components that do not correspond to items in profit classes $k \in K$ with at least ${\frac{1}{\epsilon}}+1$ $t$-included items. We use this fact to show that $\floor{\x}$ is an integral solution that is $(1-2\epsilon)$-optimal.

\begin{theorem}\label{thr:rounding}
	Let $x^*$ be an optimal solution to \eqref{eq:opt-S-rho}.
	Algorithm \ref{algo:x-opt} produces, in time $O(T+n)$, a vector $\x \in P(S,\rho)$ such that $\sum_{t \in [T]}p^T \floor{\x_{t}} \geq (1 - 2\varepsilon) \sum_{t \in [T]}p^T \xo_{t}$. 
\end{theorem}
\noindent Theorem \ref{thr:rounding} will be proved in a series of intermediate steps.
\newpage
	\begin{myclaim}\label{cl:important}
		Let $t \in [T-1]$. Then:
		\begin{enumerate}[(i)]
			\item ${\cal I}^t\subseteq{\cal I}^{t+1}$ and ${\cal E}^t\supseteq{\cal E}^{t+1}$.
			\item 	${\cal I}^{t+1}\setminus{\cal I}^{t} \subseteq  {\cal E}^{t}$.
		\end{enumerate} 
	\end{myclaim}
	
	\begin{proof}	\begin{enumerate}[(i)]
			\item Immediately from the definition.
			\item If ${\cal I}^{t+1}\setminus{\cal I}^{t}\neq \emptyset$, we deduce $ t+1 \in \T$. Let $h \in [|S|]$ be such that $t_h \leq t < t_{h-1}$, where for completeness $t_0=T+1$. By construction, the items ${\cal I}^{t+1}\setminus{\cal I}^{t}$ can only be in buckets $k: k_h \leq k < k_{h+1} + \Ce$ where either $k < k_{h+1}$ or $k \in \KK_{h+1}$ and $\rho(t,k) \leq   \frac{1}{\varepsilon} $. Hence, all items from ${\cal I}^{t+1}\setminus{\cal I}^{t}$ are $t$-excluded.
		\end{enumerate}		
		\end{proof}

	\noindent Recall that, for $t\in[T]$, $W_t:= w^T x^*_t$. 
	The proof of the following claim easily follows by construction.
	
	\begin{myclaim}\label{cl:easy-peasy}
		\begin{enumerate}[(i)] 
			\item For any $h \in [|S|]$, $t \in [t_h-1]$, $k < k_{h+1}$ and $i \in I_k$, one has $x^*_{t,i}=\x_{t,i}=0$. 
			\item\label{obs:1} For $t\in [T-1]$ and $i \in [n]$, one has $\x_{t+1,i} \geq \x_{t,i} \geq 0$. 
			
			\item\label{obs:2}  For $t\in [T]$, one has: $x^*_{t,i}=\x_{t,i}=1$ for $i \in {\cal I}^t$ and $x^*_{t,i}=\x_{t,i}=0$ for  $i \in {\cal E}^t$.
		\end{enumerate}
		
	\end{myclaim}
 
Define ${\cal F}_t:=\{i \in [n]: \ 0 < \x_{t,i} < 1 \}$ to be the set of fractional components of $\x_t$ for $t \in [T]$. Recall that Algorithm \ref{algo:x-opt} sorts items by monotonically decreasing profit/weight ratio.
For items from a given profit class $k \in [K]$, this induces the order $i(1) < i(2) < \dots $ -- i.e. by monotonically increasing weight -- since all $i(k) \in I_k$ have the same profit.

The following claim shows that $\x$ is in fact an optimal solution to $\max\{ x: x\in P(S, \rho) \}$.

\begin{myclaim}\label{cl:x-opt}
		For each $t \in [T]$, one has $w^T\x_t= w^Tx^*_t$ and $p^T\x_t= p^Tx^*_t$. 
\end{myclaim}

\begin{proof}
		We first prove the statement on the weights by induction on $t$, the basic step being trivial. Suppose it is true up to time  $t-1$. The total weight of solution $\x_t$ after step \eqref{step:add-included} is 
$$		\begin{array}{lll}
		w^T\x_{t-1} + \sum_{i \in {\cal I}^t\setminus {\cal I}^{t-1}}w_{i}(1-\x_{t-1,i}) & = &  W_{t-1} + \sum_{i \in {\cal I}^t\setminus {\cal I}^{t-1}}w_{i}(1-x^*_{t-1,i}) \\ & = & W_{t-1}+ \sum_{i \in {\cal I}^t\setminus {\cal I}^{t-1}}w_{i} \overset{(\ast)}{\leq} W_{t},
		\end{array}$$
		

		\noindent where the equations follow by induction, Claim \ref{cl:easy-peasy}.(iii), and Claim \ref{cl:important}.(ii), and $(\ast)$ follows by observing $w^T \xo_t - w^T \xo_{t-1} \geq \sum_{i \in {\cal I}^t\setminus {\cal I}^{t-1}}w_{i}$. $\x_t$ is afterwords increased until its total weight is at most $W_t$. Last, observe that $W_t$ is always achieved, since it is achieved by $x^*_t$. This concludes the proof of the first statement. 
		
		We now move to the statement on profits. Note that it immediately follows from the optimality of $x^*$ and the first part of the claim if we show that $\x$ is the solution maximizing $p^T x_t$ for all $t \in [T]$, among all  $x \in P(S,\rho)$ that satisfy $w^Tx_t=W_t$ for all $t \in[T]$.
		So let us prove the latter. Suppose by contradiction this is not the case, and let $\tilde x$ be one such solution such that $p^T \tilde x_t>p^T \x_t$ for some $t \in [T]$. Among all such $\tilde x$, take one that is lexicographically maximal, where entries are ordered $(1,1), (1,2), \dots, (1,n), (2,1) \dots, (T,n)$. Then there exists $\tau \in [T]$, $\ell \in [n]$ such that $\tilde x_{\tau,\ell}>\x_{\tau,\ell}$. Pick $\tau$ minimum such that this happens, and $\ell$ minimum for this $\tau$.  
		Using that $\x_{\tau,i}={\tilde x}_{\tau,i}$ for $i \in \I^\tau \cup {\cal E}^{\tau}$ since $\x, {\tilde x} \in P(S, \rho )$ and recalling $w^T \x_{\tau}=w^T \tilde x_{\tau}=W_\tau$ one obtains
		\begin{equation} \label{eq:resid_weigh}
		\sum_{i\in [n]\setminus ({\cal I}^{\tau}\cup {\cal E}^{\tau})} w_{i} \x_{{\tau},i} = \sum_{i\in [n] \setminus ({\cal I}^{\tau}\cup {\cal E}^{\tau})} w_{i} \tilde x_{{\tau},i} .
		\end{equation}
		
		It must be that $\x_{\tau,\ell} < 1$, since $\x_{\tau,\ell} < \tilde x_{\tau,\ell} \leq 1$, so step (c) of Algorithm \ref{algo:x-opt} in iteration $\tau$ did not change any item $\hat \ell > \ell$, i.e. $\x_{\tau,{\hat \ell}} = \x_{\tau-1,{\hat \ell}}$ for each ${\hat \ell} > \ell$. Additionally, $\ell \notin {\cal I}^\tau$ beacuse $\x_{\tau,\ell} < 1$, and $\ell \notin {\cal E}^\tau$ since otherwise $\x_{\tau,\ell} = {\tilde x}_{\tau,\ell} = 0$. Hence, $\ell \in [n]\setminus ({\cal I}^{\tau}\cup {\cal E}^{\tau})$. By moving the terms corresponding to ${\hat \ell} > \ell$ to the right-hand side, we rewrite (\ref{eq:resid_weigh}) as follows
		\begin{equation*}  
		\sum_{ \substack{{\bar \ell} \in [n] \setminus ({\cal I}^{\tau}\cup {\cal E}^{\tau}): \\ 
				{\bar \ell} \leq \ell} } w_{\bar \ell} \x_{{\tau},{\bar \ell}}  = 
		\sum_{ \substack{{\bar \ell} \in [n] \setminus ({\cal I}^{\tau}\cup {\cal E}^{\tau}): \\ 
				{\bar \ell} \leq \ell} } w_{\bar \ell} \tilde x_{{\tau},{\bar \ell}} +
		\sum_{ \substack{{\hat \ell} \in [n] \setminus ({\cal I}^{\tau}\cup {\cal E}^{\tau}): \\ 
				{\hat \ell} > \ell} } w_{\hat \ell} (\tilde x_{{\tau},{\hat \ell}} - \underbrace{\x_{{\tau},{\hat \ell}}}_{=\x_{{\tau-1},{\hat \ell}}} ) .
		\end{equation*}
		By minimality of $\tau$ one has $\tilde x_{\tau-1} \leq \x_{\tau-1}$, so $w^T \tilde x_{\tau-1}=W_{\tau-1}=w^T\x_{\tau-1}$ implies $\tilde x_{\tau-1} = \x_{\tau-1}$ and thus 
		\begin{equation} \label{eq:l_weigh} 
		\sum_{ \substack{{\bar \ell} \in [n] \setminus ({\cal I}^{\tau}\cup {\cal E}^{\tau}): \\ 
				{\bar \ell} \leq \ell} } w_{\bar \ell} \x_{{\tau},{\bar \ell}}  = 
		\sum_{ \substack{{\bar \ell} \in [n] \setminus ({\cal I}^{\tau}\cup {\cal E}^{\tau}): \\ 
				{\bar \ell} \leq \ell} } w_{\bar \ell} \tilde x_{{\tau},{\bar \ell}} +
		\overbrace{
			\sum_{ \substack{{\hat \ell} \in [n] \setminus ({\cal I}^{\tau}\cup {\cal E}^{\tau}): \\ 
					{\hat \ell} > \ell} } w_{\hat \ell} (\tilde x_{{\tau},{\hat \ell}} - \tilde x_{{\tau-1},{\hat \ell}} )
		}^{\geq 0} .
		\end{equation}
		
		%
		Note that the items in $[n]$ are ordered according to monotonically decreasing profit/weight ratio. By minimality of $\ell$ subject to $\tau$ we have that $\x_{\tau,{\bar \ell}}  \geq \tilde x_{\tau,{\bar \ell}}$ for ${\bar \ell} < \ell$. Thus combining $\x_{\tau,\ell} < \tilde x_{\tau,\ell}$ with (\ref{eq:l_weigh}) gives that there exists $\beta<\ell$ such that $\x_{\tau,\beta} >\tilde x_{\tau,{\bar \ell}}$. Then for all $\bar \tau \geq \tau$, one can perturb $\tilde x$  by increasing $\tilde x_{\bar \tau,{\beta}}$ and decreasing $\tilde x_{\bar \tau,\ell}$  while keeping $\tilde x \in P(S,\rho)$ and $w^T \tilde x_{\bar \tau}=W_{\bar \tau}$, without decreasing $p^T\tilde x_{\bar \tau}$. This contradicts the choice of $\tilde x$ being lexicographically maximal.
\end{proof}


For $t \in [T]$ define ${\cal L}_t:=\{k \in [K]: \ |I_k \cap {\cal I}^t| \geq {\frac{1}{\epsilon}}+1 \ \}$ to be the set of classes with a large number of $t$-included items.  Furthermore, for $h = 1,2, \dots, |S|$:
\begin{itemize}
	\item
		Recall that $\KK_h=\{k \in [K]: \ k_h \leq k \leq k_h + \Ce \}$ are the classes of most profitable items present in the knapsack at times $t \in [T]: \ t_h \leq t < t_{h-1}$, since by definition no item is taken from a class $k < k_h$ at those times. Also by definition $\rho(t_h, k_h) \geq 1$, so the largest profit item present in the knapsack at any time $t \in [T]: t_h \leq t < t_{h-1}$ is item $1(k_h)$. Denote its profit by $p_{max}^h$.
	\item
	    Define $\bar\KK_h:=\{k \in [K]: \ k_h + \Ce < k \}$, i.e. it is the family of the other classes for which an object may be present in the knapsack at time $t \in [T]: \ t_h \leq t < t_{h-1}$.  
\end{itemize}




\begin{myclaim}\label{cl:who-many-fract}
	Fix $t \in [T]$, $t_h \leq t < t_{h-1}$. Then, $|I_k \cap {\cal F}_t|\leq 1$ for all $k \in [K]\cup \{\infty\}$. Moreover, $|((\cup_{k \in \bar\KK_h} I_k) \cap {\cal F}_t) \setminus {\cal F}_{t_{h-1}}| \leq 1$. \end{myclaim}

\begin{proof}
	
		We show this by induction on $t$. Fix $t\geq 1$ and suppose that $|I_k \cap {\cal F}_t|\leq 1$ for all $k\in [K]\cup\{\infty\}$. By construction, for a class $k$ such that $I_k \cap {\cal F}_t=\{i_k\}$, all items $j \in I_k$ with $\bar x_{t,j}=0$ follow $i_k$ in the profit/weight order. Hence, at time $t+1$, the algorithm will not increase $\bar x_{t+1,j}$ for any $j \in I_k$ until $\bar x_{t+1,i_k}$ is set to $1$. We can repeat this argument and conclude $|I_k \cap {\cal F}_{t+1}|\leq 1$. 
		Note that this also settles the basic step $t=0$ and the case $I_k \cap {\cal F}_t=\emptyset$, concluding the proof of the first part. A similar argument settles the other statement.	
\end{proof}

\begin{myclaim}\label{cl:bar-x-fract}
Let $h \in [|S|]$, then: $p((\cup_{k \in \bar\KK_h \setminus {\cal L}_t} I_k) \cap {\cal F}_t ) \leq \epsilon \sum_{\bar{h}=h}^{|S|} p_{\max}^{\bar{h}}$, $\ \forall t: t_h \leq t < t_{h-1}$.
\end{myclaim}

\begin{proof}
		We prove the statement by induction on $h$. For $h=|S|$, let $t$ be such that $t_{|S|} \leq t < t_{|S|-1}$ and $\bar{t}=t_{|S|}-1$. We have that $(\cup_{k \in \bar\KK_{|S|}} I_k) \cap {\cal F}_{\bar{t}}=\emptyset$ so $((\cup_{k \in \bar\KK_{|S|}} I_k) \cap {\cal F}_t) \setminus {\cal F}_{\bar{t}} = (\cup_{k \in \bar\KK_{|S|}} I_k) \cap {\cal F}_t$. By using Claim \ref{cl:who-many-fract} we obtain 
		$$|(\cup_{k \in \bar\KK_{|S|} \setminus {\cal L}_t} I_k) \cap {\cal F}_t| \leq |(\cup_{k \in \bar\KK_{|S|}} I_k) \cap {\cal F}_t| \leq 1.$$ 
		The largest profit of an item in $\cup_{k \in \bar\KK_{|S|}} I_k$ is smaller than $(1+\epsilon)^{-\Ce} p_{\max}^{|S|} \leq \epsilon p_{\max}^{|S|}$ by the definition of $\bar\KK_{|S|}$ and recalling $\Ce={\log_{1+\varepsilon} \frac{1}{\varepsilon}}$. The statement follows.
		
		Assume that the statement holds for all $h$ such that $2 \leq h \leq |S|$ and prove it for $h=1$. Let $t$ such that $t_1 \leq t < t_0=T+1$ and $\bar{t}=t_1-1$. Observe that ${\cal L}_t \supseteq {\cal L}_{\bar{t}}$ and $(\cup_{k \in \bar\KK_1} I_k) \cap {\cal F}_{\bar t} \subseteq (\cup_{k \in \KK_2 \cup \bar\KK_2} I_k) \cap {\cal F}_{\bar t}$ so 
		$$(\cup_{k \in \bar\KK_1 \setminus {\cal L}_t} I_k) \cap {\cal F}_{\bar t} \subseteq (\cup_{k \in (\KK_2 \cup \bar\KK_2) \setminus {\cal L}_{\bar{t}}} I_k) \cap {\cal F}_{\bar t} = (\cup_{k \in \bar\KK_2 \setminus {\cal L}_{\bar{t}}} I_k) \cap {\cal F}_{\bar t}.$$ 
		Thus, we obtain:
		$$\begin{array}{lll}
		p((\cup_{k \in \bar\KK_1 \setminus {\cal L}_t} I_k) \cap {\cal F}_t ) & = & 
		p(((\cup_{k \in \bar\KK_1 \setminus {\cal L}_t} I_k) \cap {\cal F}_t) \setminus {\cal F}_{\bar t} ) +
		p((\cup_{k \in \bar\KK_1 \setminus {\cal L}_t} I_k) \cap {\cal F}_{\bar t} ) \\
		& \leq &  
		p(((\cup_{k \in \bar\KK_1 \setminus {\cal L}_t} I_k) \cap {\cal F}_t) \setminus {\cal F}_{\bar t} ) +
		p((\cup_{k \in \bar\KK_2 \setminus {\cal L}_{\bar{t}}} I_k) \cap {\cal F}_{\bar t}) \\ 
		& \leq & \epsilon p_{\max}^1 + \epsilon \sum_{\bar{h}=2}^{|S|} p_{\max}^{\bar{h}},
		\end{array}
		$$
		where in the last inequality we used Claim \ref{cl:who-many-fract} and the inductive hypothesis. 
\end{proof}


\noindent \textbf{Proof of Theorem \ref{thr:rounding}}. We focus on showing that, $\forall t \in [T]$:
\begin{equation} \label{eq:opt-K-approx}
\sum_{i \in [n] \setminus I_\infty} p_i \floor{\x_{t,i}} \geq \sum_{i \in [n] \setminus I_\infty} p_i \x_{t,i} - \sum_{i \in ([n] \setminus I_\infty) \cap {\cal F}_t} p_i \geq (1-\epsilon) \sum_{i \in [n] \setminus I_\infty} p_i \x_{t,i}.
\end{equation} 
The first inequality is trivial and, if $t < t_{|S|}$, so is the second, since in this case $\bar x_{t,i}=0$ for all $i \in [n]\setminus I_\infty$. Otherwise, $t$ is such that $t_h \leq t < t_{h-1}$ for some $h \in [|S|]$ with $t_0=T+1$. Observe that:
%
%
$$\begin{array}{lll}
([n] \setminus I_\infty) \cap {\cal F}_t & = & ((\cup_{k \in (\KK_h \cup \bar\KK_h) \setminus {\cal L}_t}I_k) \cap {\cal F}_t ) \cup ((\cup_{k \in (\KK_h \cup \bar\KK_h) \cap {\cal L}_t}I_k) \cap {\cal F}_t ) \\
& = & ((\cup_{k \in \bar\KK_h \setminus {\cal L}_t}I_k) \cap {\cal F}_t ) \cup ((\cup_{k \in {\cal L}_t}I_k) \cap {\cal F}_t )
\end{array}$$
For $k \in [K]$ denote the profit of $i \in I_k$ with $p^k$. We have:
\begin{equation}\label{eq:opt-max-fract}\begin{array}{lll}
\sum_{i \in ([n] \setminus I_\infty)\cap {\cal F}_t} p_i \x_{t,i} & =
& p((\cup_{k \in \bar\KK_h \setminus {\cal L}_t}I_k) \cap {\cal F}_t ) + p((\cup_{k \in {\cal L}_t}I_k) \cap {\cal F}_t ) \\ \hbox{(By Claim \ref{cl:bar-x-fract} and Claim \ref{cl:who-many-fract})}& \leq  & \epsilon \sum_{\bar{h}=h}^{|S|} p_{\max}^{\bar{h}} + \sum_{k \in {\cal L}_t} p^k.
\end{array}\end{equation}

If $k=k_{\bar h} \in {\cal L}_t $ for $\bar{h} \in [|S|]$ then $\sum_{i \in I_k}p_i \x_{t,i} \geq ({\frac{1}{\epsilon}} + 1) p^k = p_{\max}^{\bar{h}} + {\frac{1}{\epsilon}} p^k$. Together with $\rho(k_h, t_h) \geq 1 \ \forall h \in [|S|]$ and the definition of ${\cal L}_t$ this gives:
\begin{equation} \label{eq:opt-min-int}
\sum_{i \in [n] \setminus I_\infty} p_i \x_{t,i} \geq \sum_{\bar{h}=h}^{|S|} p_{\max}^{\bar{h}} +  \frac{1}{\epsilon} \sum_{k \in {\cal L}_t} p^k. 
\end{equation} 
Put together, \eqref{eq:opt-max-fract} and \eqref{eq:opt-min-int} imply \eqref{eq:opt-K-approx}. Morever, by Claim \ref{cl:who-many-fract}, $|I_\infty \cap {\cal F}_t| \leq 1$ for all $t \in [T]$ and since we are working with an $\epsilon$-well-behaved instance $p_i \leq \frac{\epsilon}{T}=\frac{\epsilon}{T}p_{\max}^1$ so $\sum_{t \in [T]} \sum_{i \in I_\infty \cap {\cal F}_t} p_i \leq \epsilon p_{\max}^1$. The last fact with \eqref{eq:opt-K-approx} and Claim \ref{cl:x-opt} gives the statement of the theorem. 
	\qedd


	

Theorem \ref{thr:IIK} now easily follows from Theorems \ref{thr:well-behaved}, \ref{thr:disjunction-iik}, and \ref{thr:rounding}. 

\smallskip 

\noindent \textbf{Proof of Theorem \ref{thr:IIK}.} Since we will need items to be sorted by profit/weight ratio, we can do this once and for all before any guessing is performed. Classical algorithms implement this in $O(n\log n)$. By Theorem \ref{thr:well-behaved}, we know we can assume that the input instance is $\varepsilon$-well-behaved, and it is enough to find a solution of profit at least $(1-2\varepsilon)$ the profit of a $1$-in solution of maximum profit -- by Theorem \ref{thr:rounding}, this is exactly vector $\floor{\x}$. 
In order to produce $\floor{\x}$, as we already sorted items by profit/weight ratio, we only need to solve the LPs associated with each choice of $S$ and $\rho$, and then run Algorithm \ref{algo:x-opt}. The number of choices of $S$ and $\rho$ are $T^{f(\varepsilon)}$, and each LP has $g(\varepsilon)O(nT)$ constraints, for appropriate functions $f$ and $g$ (see the proof of Theorem \ref{thr:disjunction-iik}). Algorithm \ref{algo:x-opt} runs in time $O(\frac{T}{\varepsilon}\log\frac{T}{\varepsilon}+n)$. The overall running time is: 
$$O(n\log n+n(n+T+ T^{f(\varepsilon)}(f_{LP}(g(\varepsilon)O(nT))+\frac{T}{\varepsilon}\log\frac{T}{\varepsilon})))=O(nT^{h(\varepsilon)}f_{LP}(n)),$$
\noindent where $f_{LP}(m)$is the time required to solve an LP with $O(m)$ variables and constraints, and $h: \R\rightarrow \setN_{\geq 1}$ is an appropriate function. \qedd


\section{Generalizations}\label{sec:ext}

Following Theorem \ref{thr:IIK}, one could ask for a PTAS for the general incremental knapsack (\ik) problem. This is the modification of \iik \ (introduced in \cite{Bienstock13}) where the objective function is $
p_{\Delta}(x):= \sum_{t \in [T]} \Delta_t \cdot p^T x_t
$,
where $\Delta_t \in \Zp$ for $t \in [T]$ can be seen as time-dependent discounts. We show here some partial results. 


\begin{corollary} \label{cor:monotone}
	There exists a PTAS-preserving reduction from IK to \iik, assuming $\Delta_t \leq \Delta_{t+1}$ for $t \in [T-1]$. Hence, the hypothesis above, IK has a PTAS.
\end{corollary}


\noindent We start by proving an auxiliary corollary.

\begin{corollary} \label{cor:maxdelta}
	There exists a strict approximation-preserving reduction from \ik \ to \iik,  assuming that the maximum discount $\Delta_{\max} := \|\Delta\|_{\infty}$ is bounded by a polynomial  
	$$g(T, n,\log\|p\|_{\infty}, \log\|w\|_{\infty}).$$ 
	In particular, under the hypothesis above, \ik \ has a PTAS.
\end{corollary}

\begin{proof} Let $\I := (n, p, w, T, b, \Delta)$ be an IK instance with 
$\Delta_{\max} \leq g(T, n,\log\|p\|_{\infty}, \log\|w\|_{\infty}).$
The corresponding instance $\I' := (n, p, w, T', b')$ of \iik \ is obtained by setting $$T':=\sum\limits_{t \in [T]} \Delta_t \quad \hbox{ and } \quad b_{t'}':=b_t, \ \forall t' \in [T']: \ \delta_t +1 \leq t' \leq \delta_t+\Delta_t,$$ 
where $\delta_t :=\sum_{\bar t < t} \Delta_{\bar t}$ for $t \in [T]$. We have that $T' \leq T \cdot g(T, n,\log\|p\|_{\infty}, \log\|w\|_{\infty})$ so the size 
	of $\I'$ is polynomial in the size of $\I$.
	
	Given an optimal solution $\xo \in \{0,1\}^{Tn}$ to $\I$, and $x' \in \{0,1\}^{T'n}$ such that $x_{t'}'=x_t$ for all $t \in [T]$ and $\delta_t +1 \leq t' \leq \delta_t+\Delta_t$, one has that $x'$ is feasible in $\I'$ so 
	\[
	\text{OPT}(\I) = p_{\Delta}(\xo) = \sum\limits_{t \in [T]} \Delta_t \cdot p^T \xo_t = \sum\limits_{t' \in [T']} p^T x'_{t'} \leq \text{OPT}(\I').
	\] 
	Let $\hat{x}$ be a $\alpha$-approximated solution to $\I'$. Define $\x \in \{0,1\}^{Tn}$ as $\x_t= \hat{x}_{\delta_{t}+\Delta_t}$ for $t \in [T]$. Then clearly $\x_t\leq \x_{t+1}$ for $t \in [T-1]$. Moreover,
	$$w^T \bar x_t = w^T \hat{x}_{\delta_{t}+\Delta_t} \leq b'_{\delta_{t}+\Delta_t} = b_t, \quad \ \forall t \in [T].$$ 
	Hence $\x$ is a feasible solution for $\I$ and 
	$$p_{\Delta}(\x)=\sum\limits_{t \in [T]} \Delta_t \cdot p^T \x_t \geq \sum\limits_{\bar{t} \in [T']} p^T \hat x_{\bar{t}}.$$ 
	Finally, one obtains:
	\begin{equation}
		\frac{p_{\Delta}(\x)}{\text{OPT}(\I)} \geq \frac{\sum_{\bar{t} \in [T']} p^T \hat x_{\bar{t}}}{\text{OPT}(\I')} \geq \alpha.
	\end{equation}	
	%
		
\end{proof}

\noindent \textbf{Proof of Corollary \ref{cor:monotone}.} Given an instance $\I$ of IK with monotonically increasing discounts, and letting $p_{\max} := \|p\|_{\infty}$, we have that the optimal solution of $\I$ is at least $\Delta_{max} \cdot p_{\max}$ since $w_i \leq b_T,\   \forall i \in [n]$, otherwise an element $i$ can be discarded from the consideration. Reduce $\I$ to an instance $\I'$ by setting $C=\frac{\varepsilon \Delta_{\max}}{Tn}$ and $\Delta_{t}'=\floor{\frac{\Delta_{t}}{C}}$. We get that $\Delta_{\max}' \leq Tn / \varepsilon$ thus satisfying the assumption of Corollary \ref{cor:maxdelta} for each fixed $\varepsilon>0$. Let $\xo$ be an optimal solution to $\I$ and $\x$ a $(1-\varepsilon)$-approximated solution to $\I'$, one has:
\vspace{-2mm}
\begin{equation*}
\begin{array}{llll}
p_{\Delta}(\x) &\geq C \cdot p_{\Delta}'(\x) \\ & \geq C \cdot (1-\varepsilon) p_{\Delta}'(\xo)\\ & \geq  (1-\varepsilon) ( p_{\Delta}(\xo) - C \sum\limits_t p^T \xo_t) \\
\ &\geq (1-\varepsilon)(p_{\Delta}(\xo) - \varepsilon \Delta_{max} \cdot p_{\max}) & \geq (1-2\varepsilon) p_{\Delta}(\xo).
\end{array}
\vspace{-7mm} 
\end{equation*} \qedd

\medskip

The proof of Corollary \ref{cor:monotone} only uses the fact that an item of the maximum profit is feasible at a time  with the highest discount. Thus its implications are broader. 

Of independent interest is the fact that there is a PTAS for the modified version of \iik \ when each item can be taken multiple times. Unlike Corollary \ref{cor:monotone}, this is not based on a reduction between problems, but on a modification on our algorithm. 

\begin{corollary} \label{cor:unbound}There is a PTAS for the following modification of \iik: in \eqref{eq:IIK}, replace $x_t \in \{0,1\}^n$ with: $x_t \in \Zp^{n}$ for $t \in [T]$; and $0 \leq x_t \leq d$ for $t \in [T]$, where we let $d \in (\Zp \cup \{\infty\})^n$ be part of the input. 
\end{corollary}

\begin{proof}

We detail the changes to be implemented to the algorithm and omit the analysis, since it closes follows that for \iik. Modify the definition of $P(S, \rho)$ as follows. Fix $h \in [|S|]$, $ k \in {\cal K}_h$ and $t \in {\cal T}_h$. As before, items in the $k$-th bucket are ordered monotonically increasing according to their weight as $I_k=\{1(k),\dots , |I_k|(k) \}$. In order to take into account item multiplicities we define $r:= r(t,k) = \max\{\bar{r}: \ \sum_{l=1}^{\bar{r}} d_{l(k)} < \rho({t,k}) \}.$ 
Replace the third, fitfth and sixth set of constraints from $P(S,\rho)$ with the following, respectively: 
\vspace{-1mm}
\begin{itemize}
	\setlength{\itemsep}{0pt}
	\item[(4')] $0 \leq x_t \leq d$;
	\item[(5')] $x_{t,i(k)}=d_{i(k)}, \ \forall i(k): i \leq r(t,k)$;
	$x_{t,(r+1)(k)}=\rho(t,k)-\sum_{l=1}^{r} d_{l(k)}$;
	\item[(6')] $x_{t,i(r+2)}=0, \dots , x_{t,i(|I_k|)}=0$ if $\rho_{t,k}\leq \frac{1}{\varepsilon}$.
\end{itemize}
For fixed $S, \rho$, call all items $i$ such that $x_{t,i}=c$ appears in $(5')$ or in $(6')$ \emph{$(t,c)$-fixed}. Note that items that are $(t,0)$-fixed correspond to items that were called $t$-excluded in \iik. Items that are $(t,c)$-fixed for some $c$ are called \emph{$t$-fixed}. Let $\bar x$ be the output of the modification of Algorithm \ref{algo:x-opt} given below.
	Again, vector $\lfloor \bar x \rfloor$ gives the required $(1-2\epsilon)$-approximated integer solution.
\begin{algorithm}[h!t]
	\caption{}
	\label{algo:x-opt2}
	\begin{algorithmic} [1]
		\medskip
		\STATE{Set $\bar{x}_{0}=\boldsymbol{0}$.}
		\medskip
		\STATE{For $t=1,\dots,T$: 
			\vspace{-3mm}
			\begin{enumerate}[(a)]
				\setlength{\itemsep}{0pt}
				\item Set $\bar{x}_t=\bar{x}_{t-1}$.
				\item\label{step:add-included2} For $i \in [n]$, if $i$ is $(t,c)$-fixed for some $c$, set $\bar x_{t,i}=c$. 
				\item While {$W_t - w^T\bar{x}_t>0$:} 
				\vspace{-3mm}
				\begin{enumerate}[(i)]
					\setlength{\itemsep}{0pt}
					\item Select the smallest $i \in [n]$ such that $i$ is not $t$-fixed and $\bar{x}_{t,i}<d_i$.
					\item Set $\bar{x}_{t,i}=\bar{x}_{t,i} + \min \{d_i-\bar{x}_{t,i}, \frac{W_t - w^T\bar{x}_t }{w_i} \}$. 
				\end{enumerate}
			\end{enumerate}
			\vspace{-5mm}}
		\end{algorithmic}
	\end{algorithm}
\end{proof}
\newpage

\noindent {\bf Acknowledgements.} We thank Andrey Kupavskii for valuable combinatorial insights on the topic. \ Yuri Faenza's research was partially supported by the SNSF \emph{Ambizione} fellowship PZ00P2$\_$154779 \emph{Tight formulations of $0/1$ problems}. Some of the work was done when Igor Malinovi\'c visited Columbia University partially funded by a gift from the SNSF.

\bibliographystyle{plain}
\bibliography{citation}


\appendix

\section*{Appendix}

\section{Notation}

We refer to \cite{Schrijver03} for basic definitions and facts on approximation algorithms and polytopes. Given an integer $k$, we write $[k]:=\{1,\dots,k\}$ and $[k]_0:=[k]\cup \{0\}$. Given a polyhedron $Q\subseteq \R^n$, a \emph{relaxation} $P\subseteq \R^n$ is a polyhedron such that $Q\subseteq P$ and the integer points in $P$ and $Q$ coincide. The \emph{size} of a polyhedron is the minimum number of facets in an extended formulation for it, which is well-known to coincide with the minimum number of inequalities in any linear description of the extended formulation. 

\section{Background on disjunctive programming}\label{app:disj-progr}

Introduced by Balas \cite{Balas} in the 70s, it is based on ``covering'' the set by a small number of pieces which admit an easy linear description. More formally, given a  set $Q \subseteq \mathbb{Z}^n$ we first find a collection $\{Q_j \}_{j \in [m]}$ such that $Q=\cup_{j \in [m]}Q_j$.
If there exist polyhedra $P_j, j \in [m]$ with bounded integrality gap and $P_j \cap \mathbb{Z}^n = Q_j$, then $P:=\conv{\cup_{j \in [m]}P_j}$ is a relaxation of $\conv Q$ of with the same guarantee on the integrality gap. Moreover, one can describe $P$ with (roughly) as many inequalities as the sum of the inequalities needed to describe the $P_j$. 
A variety of benchmarks of mixed integer linear programs (MILPs) have shown the improved performances of branch-and-cut algorithms by efficiently generated disjunctive cuts \cite{Balas09}.  Branch-and-bound algorithms for solving MILP also implicitly use disjunctive programming. The branching strategy based on thin directions that come from the Lenstra's algorithm for integer programming in fixed dimension has shown good results in practice for decomposable knapsack problems \cite{Krishnamoorthy09}. For further applications of disjunctive cuts in both linear and non-linear mixed integer settings see \cite{Belotti11}.

\section{IIK, MKP, and UFP}\label{app:ufp}
 A special case of GAP where profits and weights of items do not change over the set of bins is called the multiple knapsack problem (MKP). MKP is strongly NP-complete as well as \iik \ and has an LP-based efficient PTAS (EPTAS) \cite{Jansen12}. Both the scheme in \cite{Jansen12} and the one we present here are based on reducing the number of possible profit classes and knapsack capacities, and then guessing the most profitable items in each class. However, the way this is performed is very different. The key ingredient of the approximation schemes so far developed for MKP is a ``shifting trick''. In rounding a fractional LP solution it redistributes and mixes together items from different buckets. Applying this technique to \iik \ would easily violate the monotonicity constraint, i.e. $x_{t, i} \leq x_{t+1, i} $ where $x_{t,i}$ indicates whether an item $i$ is present in the knapsack at time $t$. This highlights a significant difference between the problems: the ordering of the bins is irrelevant for MKP while it is crucial for \iik.

In \ufp \ one is given a path $P=(V,E)$ with edge capacities $b: E \rightarrow\R_{>0}$ and a set of tasks (i.e. sub-paths) $[n]$ with profits $p: [n] \rightarrow \R_{>0}$ and weights $w: [n] \rightarrow \R_{>0}$ and, for each task $\pi \in [n]$, its starting point and ending nodes $u(\pi), v(\pi) \in V$. The goal is to select a set $S \subseteq [n]$ of maximum profit such that, for each $e \in E$, the set of tasks in $S$ containing $e$ has total weight at most $b_e$. One might like to rephrase \iik \ in this framework mapping times to nodes, parameters $b_t$ to edge capacities, and the insertion of item $i$ at time $t$ with an appropriate path $\pi(t,i)$. However, we would need to introduce another set of constraints that for each item $i$ at most one task $\pi=(i,t)$ is taken. This would be a more restrictive setting then \ufp. The best known approximation for \ufp \ is $2 + \epsilon$ \cite{Anagnos14}. When all tasks share a common edge, there is a PTAS \cite{Grandoni17} based on a ``sparsification'' lemma introduced in \cite{Batra15} which, roughly speaking, considers guessing $1/ \epsilon$ ``locally large'' tasks in the optimal solution for each $e \in E$ and by this making the computation of ``locally small'' tasks easier. In our approach for solving \iik \ we perform a kind of sparsification in Section \ref{sec:reduction} by reducing the number of times and different profits to be taken into consideration. At that point, the number of possible time/profit combinations is still too large to be able to guess a constant fraction of the highest profit items per each time. Thus, we introduce an additional pattern enumeration in Section \ref{sec:iikdisj} which follows the evolution of the highest-profit item in an optimal solution to an \iik \ instance. This pattern, -- that we call ''stairway``, see Section \ref{sec:iikdisj} -- is specific for \iik, 
and fundamental for describing its dynamic nature (while 
the set of edges for \ufp \ is fixed). Once the stairway is fixed we can identify and distinguish between locally large and small items. This is the main  difference between our approach here and the techniques used for \ufp \ and related problems \cite{Anagnos14,Batra15,Grandoni17}, or the techniques used in other works on \iik \ \cite{Bienstock13,Ye16}.

\end{document}